\documentclass[conference,twocolumn]{IEEEtran}
\usepackage{times,amsmath,epsfig}
\usepackage{latexsym,amssymb}
\usepackage{graphicx,amssymb,amstext,amsmath}
\usepackage{setspace}
\usepackage{amsmath}
\usepackage{amsthm}
\usepackage{eucal}
\usepackage{amssymb}
\usepackage{mathrsfs}
\newtheorem{theorem}{Theorem}

\newtheorem{lemma}[theorem]{Lemma}%[chapter]
\begin{document}

\sloppy

%% Paper Title
%% You can use linebreaks \\ within to get better formatting as
%% desired. 
\title{To Obtain or not to Obtain CSI in the Presence of Hybrid Adversary} 
%% Author names and affiliations:
%%
%% Avoiding spaces at the end of the author lines is not a problem with
%% conference papers because we don't use \thanks or \IEEEmembership.
%%
%% For several authors with only one affiliation:
%%
% \author{
%   \IEEEauthorblockN{Hui-Ting Chang and Stefan M.~Moser}
%   \IEEEauthorblockA{Department of Electrical and Computer Engineering\\
%     National Chiao Tung University (NCTU)\\
%     Hsinchu, Taiwan\\
%     Email: \{email-of-hui-ting,email-of-stefan\}@ieee.org} 
% }
%%
%% For up to three affiliations:
%%
\author{
  \IEEEauthorblockN{Y.~Ozan Basciftci}
  \IEEEauthorblockA{Dep. of Electrical \& Computer Eng.\\
    The Ohio State University\\
    Columbus, Ohio, USA\\
    Email: basciftci.1@osu.edu} 
  \and
  \IEEEauthorblockN{C.~Emre Koksal}
  \IEEEauthorblockA{Dep. of Electrical \& Computer Eng.\\
    The Ohio State University\\
    Columbus, Ohio, USA\\
    Email: koksal@ece.osu.edu }
  \and
  \IEEEauthorblockN{Fusun Ozguner}
  \IEEEauthorblockA{Dep. of Electrical \& Computer Eng.\\ 
    The Ohio State University\\
    Columbus, Ohio, USA\\
    Email: ozguner@ece.osu.edu}
}
%%
%% For over three affiliations, or if they all won't fit within the width
%% of the page, use this alternative format:
%%
% \author{
%   \IEEEauthorblockN{
%     Michael Shell\IEEEauthorrefmark{1},
%     Homer Simpson\IEEEauthorrefmark{2},
%     James Kirk\IEEEauthorrefmark{3}, 
%     Montgomery Scott\IEEEauthorrefmark{3} and
%     Eldon Tyrell\IEEEauthorrefmark{4}}
%   \IEEEauthorblockA{
%     \IEEEauthorrefmark{1}School of Electrical and Computer Engineering\\
%     Georgia Institute of Technology, Atlanta, Georgia 30332--0250\\ 
%     Email: see http://www.michaelshell.org/contact.html}
%   \IEEEauthorblockA{
%     \IEEEauthorrefmark{2}Twentieth Century Fox, Springfield, USA\\
%     Email: homer@thesimpsons.com}
%   \IEEEauthorblockA{
%     \IEEEauthorrefmark{3}Starfleet Academy, San Francisco, California 96678-2391\\
%     Telephone: (800) 555--1212, Fax: (888) 555--1212}
%   \IEEEauthorblockA{
%     \IEEEauthorrefmark{4}Tyrell Inc., 123 Replicant Street, Los Angeles, California 90210--4321}
% }

%% Use for special paper notices

%% To balance the two columns, you should reduce the text-height of
%% the last page using the following command:
%%%%%%%%%%%%%%%%%%%%%%%%%%%%%%%%%%%%%%%%%%%%%%%%%%%%%%%%%%%%%%%%%%%%%
%\addtolength{\textheight}{-9.35cm}
%%%%%%%%%%%%%%%%%%%%%%%%%%%%%%%%%%%%%%%%%%%%%%%%%%%%%%%%%%%%%%%%%%%%%
%% with an appropriate value. This command must be place on the second
%% last page, i.e., for a one-page abstract here, for a two-page
%% abstract right after the \maketitle command.

%% Create the title:
\maketitle

%% Abstract: 
%% For the final version of the accepted paper, please make sure you
%% remove the comment "THIS PAPER IS ELIGIBLE FOR THE STUDENT PAPER
%% AWARD."
%%
\begin{abstract}
We consider the wiretap channel model under the presence of a hybrid, half duplex adversary that is capable of either jamming or eavesdropping at a given time. We analyzed the achievable rates under a variety of scenarios involving different methods for obtaining transmitter CSI.
Each method provides a different grade of information, not only to the transmitter on the main channel, but also to the adversary on all channels. Our analysis shows that main CSI is more valuable for the adversary than the jamming CSI in both delay-limited and ergodic scenarios. Similarly, in certain cases under the ergodic scenario, 
interestingly, no CSI may lead to higher achievable secrecy rates than with CSI.
\end{abstract}
\section{Introduction}
Information theoretic security has received a significant attention recently. One mainstream direction has been on the wireless transmission of confidential messages from a source to a destination, in the presence of internal and/or external eavesdroppers. Toward achieving that goal, the communicating pair exploits the stochasticity and the asymmetry of wireless channels between the communicating pair and the eavesdroppers. A stochastic encoder at the transmitter makes use of the available channel state information (CSI) in a way for the mutual information leaked to the adversaries remain arbitrarily small. It is designed in a way that, even when the adversaries have access to the full CSI of the main channel, i.e., between the transmitter and the receiver as well as the eavesdropper channel, i.e., between the transmitter and itself, it still will obtain an arbitrarily low rate of information on the message. Likewise, the adversary relies on CSI to make decisions. For instance, a half-duplex hybrid adversary, capable of jamming or eavesdropping at a given time (but not both simultaneously) decides between jamming vs. eavesdropping, based on the available CSI.
 
The assumption that the adversaries have full CSI of all channels is typical in the literature~\cite{liang2008}-\cite{xang2009}. While this assumption leads to robust systems in terms of providing security as it makes no assumptions on the adversaries, it can be too conservative in some cases. For example, to obtain main CSI, an adversary relies on the same resource as the transmitter: feedback from the legitimate receiver. Hence, from the perspective of the receiver, there is a tradeoff between revealing CSI and keeping it secret: If the legitimate receiver chooses not to reveal CSI, it will sacrifice some achievable rate of reliable communication, but the adversary will have to make decisions under a higher uncertainty. In this paper, we study the tradeoffs involving obtaining CSI. We ask the questions, should a legitimate pair obtain CSI and if so, what should be their strategy?

To that end, we focus on the system depicted in Figure~\ref{fig:sysmdl}. We assume all three channels to be block fading. In each block, based on the available CSI, the half-duplex adversary can choose to do jamming at a fixed transmission power or eavesdropping, but not both. Our objective is to maximize the rate of reliable communication over the main channel, subject to full equivocation~\cite{wyner1975} (weak secrecy) at the adversary. The adversary can follow an arbitrary strategy in its choice of jamming vs. eavesdropping at any given block. In the case in which the receiver feeds back main CSI, it may do so in two ways: directly by sending back the exact state of the channel or by sending reverse pilots, trying to exploit channel reciprocity (similar to~\cite{marteza2006}). While the former method completely reveals the main CSI, it eliminates the possibility of the adversary to learn the jammer CSI. On the other hand, while the pilot feedback successfully hides the main CSI, it enables the adversary to estimate the jammer CSI. In terms of the secrecy encoding strategies, we address the possibilities under two general scenarios: the ergodic and the delay-limited. In the former case, one message is encoded across infinitely many blocks and in the latter case, a separate message is encoded over each block, to be decoded immediately. Thus, in the delay-limited scenario, we also impose an additional probabilistic constraint on the decoding and secrecy outage events.

In the delay limited scenario, we show that by revealing the main CSI, the receiver achieves a higher secrecy rate under the outage constraint, compared to transmission with no CSI. Furthermore, we show that main CSI is more valuable for the adversary than the jamming CSI in both delay-limited and ergodic scenarios. In the ergodic scenario, we observe that the transmitter may not need the CSI to achieve higher secrecy rates.

There is a recent research interest on hybrid adversaris. In~\cite{mukher2013}, the authors formulate the MIMO wiretap channel as a two player zero-sum game in which the payoff function is an achievable ergodic secrecy rate. The strategy of the transmitter is to send the message in a full power or to utilize some of the available power to produce artificial noise. The conditions under which pure Nash equilibrium exists are studied. In~\cite{amariucai2012}, the authors consider fast fading main and eavesdropper channels and static jammer channel. Under this channel configuration, they propose a novel encoding scheme which is called block-Markov Wyner secrey encoding. In~\cite{zhou2012}, the authors introduce a pilot contamination attack in which the adversary jams during the reverse training phase to prevent the transmitter from estimating the main CSI correctly.  As a result, the transmitter incorrectly designs precoder which will increase the signal strength at the adversary that eavesdrops the main channel during the data transmission phase.

The rest of this paper is organized as follows. In Section
\ref{chap:system}, we first describe the system model. We then explain the channel model and CSI feedback models. At the end of the section, we explain the problem formulations. In Section~\ref{chap:result}, we present the results for both delay limited and ergodic scenarios. In Section~\ref{chap:numeric}, we present our numerical results and conclude the paper in Section~\ref{chap:conc}.
\begin{figure}[t]
   \centering
   \includegraphics[width=0.4\textwidth]{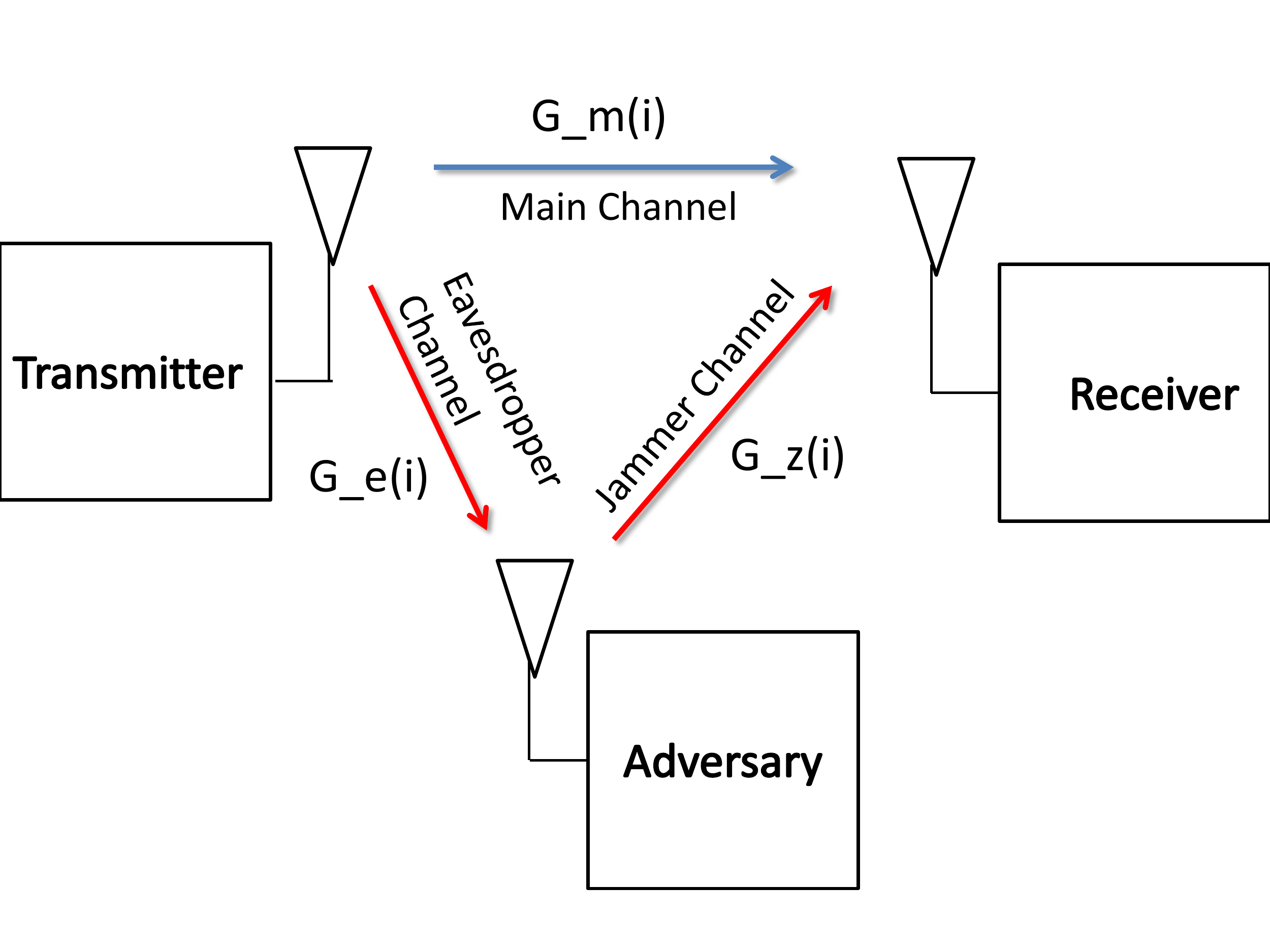}
   % where an .eps filename suffix will be assumed under latex,
   % and a .pdf suffix will be assumed for pdflatex
   \caption{System Model}
   \label{fig:sysmdl}
 \end{figure}
\section{System Model and Problem Formulation}
\label{chap:system}
\subsection{Channel Model}
\label{channelmodel}
In this paper, we focus on a block fading channel model. Time is divided into discrete blocks and there are $N$ channel uses in each block. Channel state is assumed to be constant within a block and varies randomly from one block to the next. We assume all parties are half-duplex, thus the adversary can not jam and eavesdrop simultaneously.%in the same block.

The observed signals at the legitimate receiver and the adversary in $i$-th block are as follows:
\begin{align}
Y^N(i)& = G_m(i)X^N(i)+ G_z(i)S_j^N(i) I_J(i)+ S_m^N(i)\\
Z^N(i)& = G_e(i)X^N(i)(1-I_J(i)) + S_e^N(i) 
\end{align}
\noindent where $X^N(i)$  is the transmitted signal, %with distribution $\mathcal{CN}(\mathbf{0},PI_{N\times N})$,
$P$ is the transmission power, 
$Y^N(i)$  is the signal received by the legitimate receiver,  $Z^N(i)$ is the signal received by the adversary, $S_j^N(i)$, $S_m^N(i)$, and $S_e^N(i)$ are noise vectors distributed as complex Gaussian, $\mathcal{CN}(\mathbf{0},P_jI_{N\times N})$, $\mathcal{CN}(\mathbf{0},I_{N\times N})$, and $\mathcal{CN}(\mathbf{0},
I_{N\times N})$, respectively, and $P_j$ is the jamming power. Indicator function $I_J(i)=1$, if the adversary is in a jamming state in the $i$-th block; otherwise  $I_J(i)=0$. Channel gains, $G_m(i)$, $G_e(i)$, and $G_z(i)$ are defined to be the independent complex gains of transmitter-to-receiver channel, transmitter-to-adversary channel, and adversary-to-receiver channel, respectively (as illustrated in Figure~\ref{fig:sysmdl}). Associated power gains are denoted with $H_m(i)= \lvert G_m(i)\rvert^2$, $H_e(i)= \lvert G_e(i)\rvert^2$, and $H_z(i)= \lvert G_z(i)\rvert^2$. We assume that channel reciprocity principle is valid for all channels, i.e., reverse channels and forward channels have identical gains. We also assume that joint probability density function of instantaneous power gains, $f_\mathbf{H}(\mathbf{h})$, where $\mathbf{H} =\left[ H_m(\cdot)\;H_e(\cdot)\;H_z(\cdot)\right]$, is well defined and known by all entities.

\subsection{Methods for Obtaining CSI}
The legitimate pair may choose to obtain main CSI or communicate without it. We call the latter strategy the {\em no CSI case}. If they choose to obtain main CSI, at the beginning of each time block, the transmitter sends training symbols to the legitimate receiver. In this paper, we ignore the overhead associated with this training process. We assume
%that the adversary knows the training symbols and we also assume 
that, using the training symbols sent at the beginning of block $i$, the legitimate receiver obtains perfect knowledge of $G_m(i)$ and the adversary obtains perfect knowledge of $G_e(i)$.

Once the receiver observes main CSI, it uses two possible methods for feeding back this information. The first one is directly feeding back the observed channel state: the value of $G_m(i)$ is encoded at the receiver and sent to the transmitter in a feedback packet. Thus, we call this feedback method the {\em packet feedback}. We assume that the legitimate receiver and the adversary both decode this packet successfully and learn $G_m(i)$.
The second method is using pilot based CSI feedback in which the receiver sends training symbols to the transmitter. We call this second method the {\em pilot feedback}. We assume that by using these reverse training symbols, the legitimate transmitter obtains perfect knowledge of $G_m(i)$ and the adversary obtains the perfect knowledge of $G_z(i)$. Thus, in the first method, the adversary obtains the knowledge of $G_m(i)$, but not $G_z(i)$, whereas the reverse is true in the second method.

\subsection{Adversary Model}
The goal of the adversary is to minimize the achieved secrecy rate. The strategy space of the adversary in each block is binary: jamming or eavesdropping. %We consider a realistic scenario in which
The transmitter does not observe the strategy of the adversary in any given block, whereas we assume that the adversary knows the strategy of the transmitter a priori.

%We assume that the adversary does not change its strategy throughout  a block. 
From one block to the next, the adversary chooses its strategy based on the transmitter's strategy and the obtained channel power gains\footnote{The realizations of the random variables are represented by lower case letters in the sequel.}, i.e., $h_e(i)$ for the no CSI case, $h_e(i)$, $h_m(i)$ for the packet feedback case, and $h_e(i)$, $h_z(i)$ for the pilot feedback case. We denote the vector of channel power gains observed by the adversary at the beginning of the $i$th block with $\mathbf{h}_{A}(i)$.

The secrecy level of a transmitted message is measured by the equivocation rate at the adversary. The equivocation rate at the adversary is defined as the entropy of the transmitted message conditioned on the channel output and the available CSI at the adversary. If the equivocation rate is equal to the secrecy rate, the message is said to be transmitted with {\em perfect secrecy}.

\subsection{Encoding of Information}
\label{subsec:encoding}
In our system, we consider two levels of encoding. At the higher level, secrecy is realized Wyner code introduced in~\cite{wyner1975}. There, $C_s(R_m,R_s,NM)$ represents a Wyner code of size $2^{NMR_m}$ that bears a confidential message set $W_s= \{1,2,\dots,2^{NMR_s}\}$, where $NM$ is the codeword length in number of bits. A message, $w_s\in W_s$ is mapped to $NMR_m$ bits by a secrecy encoder~\cite{wyner1975} and these $NMR_m$ bits are then mapped to channel encoded bits at the lower level of encoding using a sequence of codes, $C_i(2^{NR(h_m(i))},N)$, one for each block $i,\ 1\leq i \leq M$. Here, the code rate, $R(h_m(i))$, is chosen based on the main CSI, obtained at the transmitter. The sequence of codewords is denoted with $x^{NM} \in \mathcal{X}^{NM}$, and the decoder, $\phi(\cdot)$  maps the received sequence, $Y^{NM}$ to $\hat{w}\in \mathcal{W}$. The average error probability of the sequence of codes $\{C_i\}$ is denoted with the associated sequence $P_e^{NM}$.
%\begin{equation}
%P_e^n = \sum_{w\in\mathcal{W}} \frac{1}{2^{nR_s}} Pr(\phi(Y^n)\neq w| w \text{ was  sent})
%\end{equation}
In this paper, we focus on two scenarios as to how secrecy encoding and channel encoding are applied: delay limited and ergodic. 

\subsubsection{Delay-Limited Scenario}
In the delay-limited scenario, the transmitter encodes a separate secret message, $W_s(i)$ in each block $i$. Consequently, we use a separate secrecy encoder $C_s(i)=C_s(R_m(i),R_s,N)$ for each block, where $R_m(i)$ is chosen to be identical to $R(h_m(i))$ to meet the channel rate. The channel encoder merely maps these $NR_m(i)$ bits to $N$ Gaussian random variables, $X^N$, forming a Gaussian codebook $C_i(2^{NR(h_m(i)},N)$. For the delay-limited scenario, we define the secrecy outage and connection outage events~\cite{xang2009} as:
\begin{align}
\frac{I(X^N;Z^N|h_e(i))}{N}&> R(h_m(i))-R_s \quad \text{and} \\
\frac{I(X^N;Y^N|h_m(i))}{N} &< R(h_m(i)),
\end{align}
respectively, where $X^N\sim\mathcal{CN}(\mathbf{0},PI_{N\times N})$
%Note that the secrecy
%outage probability can be considered as an upper bound on
%the probability of unsecured packets.
\subsubsection{Ergodic Scenario}
\label{chap:ergodicformulation}
In the ergodic scenario, the transmitter has one secret message $W_s$ and encodes it over $M$ blocks using the $C_s(R_m,R_s, NM)$ encoder. In the ergodic scenario, secrecy rate $R_s$ is said to be achievable if, for any $\epsilon > 0$, there exists sequence of channel codes $\{C_i\}$ for which the following are satisfied:
\begin{align}
&P_e^{NM} \leq \epsilon \label{cond1}\\
&\frac{1}{MN} H(W_s| Z^{MN}, \mathbf{h}_{A}^M) \geq R_s-\epsilon \label{cond2}
\end{align}
for sufficiently large N and M and for any $\mathbf{h}_{A}^M \in \mathcal{A}_M$ such that $P[\mathcal{A}_M]=1$.
Here, we consider two possible channel encoding strategies:

\noindent {\bf 1. Encoding across blocks:} We use this strategy in the no main CSI case. In this strategy, the $NMR_m$ bits at the output of the secrecy encoder is channel encoded via a single $C(2^{NMR_m},NM)$ Gaussian codebook and the $NM$ symbols are transmitted over the channel over $M$ blocks.

\noindent {\bf 2. Block-by-block encoding:} We use this strategy when the main CSI is available. To utilize the main channel knowledge, the transmitter chooses some $h_z^{*}$ and encodes the information using a Gaussian codebook with the rate $R(h_m(i))=\log\left(1+\frac{Ph_m(i)}{1+P_jh_z^{*}}\right)$ over block $i$. With this choice, the transmitted codeword will not be decoded successfully if the adversary is in the jamming state and $h_z(i) > h_z^{*}$. To handle this possibility, we use a plain ARQ strategy in each block (similar to~\cite{rezki2012}). Transmissions that receive a negative acknowledgement (NAK) are retransmitted until they are decoded successfully.
\subsection{Problem Formulation}
\label{sec:formulation}
One can notice that, in both the delay-limited and the ergodic scenarios, we use a constant secrecy rate\footnote{For the delay-limited case, it is constant over the entire sequence of Wyner codes.} $R_s$. The goal of the transmitter is to maximize the secrecy rate, $R_s$ over the strategy space of secrecy encoding and channel encoding rates. To that end, we consider the worst case scenario, in which the adversary perfectly tracks the strategy of the transmitter in each block. Furthermore, since the transmitter does not know the real state of the adversary, the strategy pair of the transmitter should satisfy the constraint for any arbitrary strategy of the adversary. Thus, we choose the secrecy rate using:
\begin{equation}
R_s^*=\max_{R_m(i),R_s} \min_{I_{J}(i)} R_s \label{problem}
\end{equation}
for all $i$ in both the delay-limited and the ergodic scenarios, subject to the following outage constraint in the delay-limited scenario only:
\begin{align}
\lim_{M\to \infty} \frac{1}{M}\sum_{i=1}^M I_J (i) I_C(i) +(1-I_J (i)) I_S(i) I_C(i) \geq \alpha 
\label{Threshold}
\end{align}
with probability 1, where $I_C(i)$ and $I_S(i)$ are indicator functions that take on a value $0$ in case of a connection and a secrecy outage, respectively. We evaluate $R_s^{*}$ under the no CSI, packet feedback, and pilot feedback cases. Note that the constraint enforces that the
fraction of packets that are not in the both secrecy outage and
the connection outage should be larger than a threshold as the
number of blocks, $M$, goes to infinity.
 
Note that, the reason why we focus on the specific encoding strategies specified under delay limited and ergodic scenarios in Section~\ref{subsec:encoding} is that, the solution, $R_S^*$, of the maximin problem stated in (\ref{problem}) is unknown~\cite{liang2007}.  
\section{Results}
\label{chap:result}
\subsection{Delay Limited Scenario}
The outage constraint in~(\ref{Threshold})  is formulated as
\allowdisplaybreaks
\begin{align}
&C =\lim_{M\to \infty} \frac{1}{M}\sum_{i=1}^M I_J (i) I_{\log\left(1+\frac{PH_m(i)}{1+P_j H_z(i)}\right) \geq R(H_m(i))} \nonumber \\
&+(1-I_J (i)) I_{\log\left(1+PH_e(i)\right) \leq R(H_m(i))-R_s}\nonumber\\
&\qquad\qquad\qquad\qquad\qquad\times I_{\log\left(1+PH_m(i)\right) \geq R(H_m(i))}
\label{constraint2}
\end{align}
Note that in (\ref{constraint2}), the secrecy outage event is represented with $\log\left(1+Ph_e(i)\right) > R(h_m(i))-R_s$ and the connection outage events are given as $\log\left(1+\frac{Ph_m(i)}{1+P_j h_z(i)}\right) < R(h_m(i))$ and $\log\left(1+Ph_m(i)\right) < R(h_m(i))$.
%With Constraint~(\ref{constraint2}), the solution to Problem~(\ref{problem}) leads to the following ordering of the achievable rate with respect to the type of feedback.
\begin{theorem}
For the delay limited scenario, the solution of (\ref{problem}), subject to constraint~(\ref{constraint2}) leads to the following ordering of the achievable rate with respect to the type of feedback:
\begin{equation}
R_s ^{\text{No CSI}} \leq R_s ^{\text{Packet feedback}} \leq R_s^{\text{Pilot Feedback}}
\end{equation}
\label{thm:main}
\end{theorem}
The above theorem implies that main CSI, which is obtained with packet feedback is more valuable for the adversary than the jammer CSI, which is obtained with pilot feedback.
%This is due to the fact that the strategy of the transmitter, ($R_s, R(h_m)$) depends on the main channel power gain, $h_m$.
We now give an outline for the proof. The details of the proof can be found in Section~\ref{chap:proof}.
\begin{proof}[The Outline of Proof of Theorem~\ref{thm:main}] The basic idea is to compare the feasible set of the problem (\ref{problem})  for three cases. The feasible set is defined as 
\begin{equation}
\mathcal{F} = \{(R_s, R(\cdot)) : C \geq \alpha \text{ w.p. 1,   } \forall A_p\}
\end{equation}
\noindent where $A_p$ is the set of channel power gains, $h_A(i)$, such that the adversary is in the jamming state if $h_A(i)\in A_p$. The equivalent form of (\ref{problem}) is as follows: $R_s^{*} = \max_{(R_s, R(\cdot))\in\mathcal{F}} R_s \label{equivalent1}$. We observe that the solution to (\ref{problem}) is directly related to the size of the feasible set. The strategy pair $(R_s, R(\cdot))$ is the element of the feasible set, $\mathcal{F}$ if  $C_{\text{min}}=\min_{A_p}C[(R_s, R(\cdot))] \geq \alpha\text{ w.p. 1}$. We can write $C_{\text{min}}$ for the no CSI and packet feedback cases as follows:
\begin{align}
&C_{min}^{\text{Packet Feedback}}=\nonumber\\
& P\left[R_s +\log\left(1+PH_e\right)\leq R(H_m)\leq \log\left(1+\frac{PH_m}{1+P_j H_z}\right)\right]\nonumber\\
&C_{min}^{\text{No CSI}}=\nonumber\\
&P\left[R_s +\log\left(1+PH_e\right)\leq R\leq \log\left(1+\frac{PH_m}{1+P_j H_z}\right)\right]\nonumber
\end{align}
We can observe that $F^{\text{No CSI}} \subset F^{\text{Packet Feedback}}$ so we have $R_s ^{\text{No CSI}} \leq R_s ^{\text{Packet feedback}}$. By assuming the adversary is full-duplex, we find a lower bound, $C^{\text{lower bound}}$ for $C_{min}$ at CSI feedback cases. Then, we observe $C^{\text{lower bound}}=C_{min}^{\text{Packet Feedback}}$ which shows $C_{min}^{\text{Packet Feedback}}\leq C_{min}^{\text{Pilot Feedback}}$ and this concludes the proof.
\end{proof}
\subsection{Ergodic Scenario}
We first present  a secrecy rate that is achievable under the no CSI case.
\begin{theorem}\label{thm:nocsi}
The achieved secrecy under no CSI is:
\begin{align}
&R_s^{\text{No CSI}}= \nonumber\\ &\left[\mathbf{E}\left[\log\left(1+\frac{PH_m}{1+P_jH_z}\right)-\log\left(1+PH_e\right)\right]\right]^{+}\label{nofeedback}
\end{align}
\end{theorem}
The proof of (\ref{nofeedback}) can be found at Section~\ref{section:nocsi}. %We omit the proof due to the space limitations.
%Rate $R_m$ is selected as $\mathbf{E}\left[\log\left(1+\frac{PH_m}{1+P_jH_z}\right)\right]$. 
%Next, we give a secrecy rate that is achievable with packet feedback. 
%\begin{theorem}\cite{rezki2012}
%The achieved secrecy rate with block-by-block encoding strategy is as follows:
%\begin{align}
%&R_s^{\text{Packet Feedback}}=\nonumber\\ 
%& \max_{h_z^{*}}\mathbf{E}\left[\log\left(1+\frac{PH_m}{1+P_jh_z^{*}}\right)-\log(1+PH_e)\right]^+\nonumber\\
%&\qquad\qquad\qquad\qquad\qquad\qquad\qquad\quad\times P[H_z \leq h_z^{*}]^2
%\end{align}
%\end{theorem}
Next, we present an upper bound for the secrecy rates achieved with the block-by-block encoding strategy.% explained in Section~\ref{chap:ergodicformulation}.
\begin{theorem}
\label{upperbound}
Under the block-by-block encoding strategy, the achievable secrecy rate is upper bounded by
\begin{multline}
\label{eq:block_upper_bnd}
R_s^+ = \mathbf{E}\left[\log\left(1+\frac{PH_m}{1+P_jH_z}\right)-\log(1+PH_e)\right]^+ \\
\times P[H_z \leq h_z^*] 
\end{multline}
\end{theorem}
To find this upper bound, we employ the following strategy. When the transmitter receives a NAK signal, the transmitter sends an independent group of bits on the next block instead of retransmitting the previous packet~\cite{yara2011}. In~\cite{yara2011}, the authors use this scheme for the secret key sharing. The crucial observation is that an upper bound for the achievable secret key rate is also an upper bound for the achievable secrecy rates. The details of the proof can be found at Section~\ref{chap:upperbound}. Note that with the plain ARQ strategy described in Section~\ref{chap:ergodicformulation}, the achievable rate is identical to the expression provided in (\ref{eq:block_upper_bnd}), with $P[H_z \leq h_z^*]$ replaced with  $(P[H_z \leq h_z^*])^2$, and $H_z$ in the expectation term replaced with $h_z^{*}$ as shown in~\cite{rezki2012}.
%If $P\left[\frac{H_m}{\sigma^2+P_jH_z} \geq \frac{H_e}{\sigma^2}\right]=1 $ and $h_z^{*} \geq E[H_z]$, no CSI feedback outperforms the packet based CSI feedback. 

By comparing the upper bound given in Theorem~\ref{upperbound}, we gain understanding on the performance of no CSI case. In particular, whenever the achievable rate with the no CSI case exceeds this bound, we know for sure that encoding across blocks (no CSI) is preferable over block-by-block encoding with CSI.
%In Section~\ref{chap:numeric}, we observe that $R_s^{\text{No CSI}}$ can be higher than $R_s^+$, which shows that the CSI feedback is not required. 
The main difference between block-by-block encoding and encoding across blocks is that, in the former, the unsuccessfully received packets are discarded, whereas in the no CSI case, all the information received by the receiver is used to decode the message.
The set of parameters for which this is the case is illustrated in Section~\ref{chap:numeric} in an example.

One can also write a general relationship between the performance with packet feedback and pilot feedback:
\begin{theorem}\label{comparison}
Any secrecy rate achievable with packet feedback strategy is also achievable with pilot feedback strategy.
\end{theorem}
Theorem~\ref{comparison}, shows that main CSI is more valuable for the adversary than the jamming CSI in the ergodic scenario, as was the case in the delay limited scenario.
%We directly use the definition for the achievable secrecy rate and the independence of main and eavesdropper channels for the proof.
The proof can be found at Section~\ref{chap:comparison}
\section{Numerical Evaluation}
\label{chap:numeric}
We first analyze the delay-limited case. We assume that both main and eavesdropper channels are characterized by block Rayleigh fading,
where the main channel and eavesdropper channel power
gains follow an exponential distribution with a mean 10 and
1, respectively\footnote{Such a difference may occur in the cellular setting, when the receiver is a base station with many antennas or in a wireless LAN setting, where the receiver is located at a favorable position for reception, compared to an external adversary.}. We also assume the jamming channel does not experience fading, where power gain is equal to 1. The transmission power, $P$, and the jamming power, $P_j$ are identical and chosen to be $1$. In Figure~\ref{fig:limited}, we plot the secrecy rate, $R_s$, as a function of the outage constraint threshold, $\alpha$ under the no CSI, the packet feedback, and the pilot feedback cases. The achievable rates in Figure~\ref{fig:ergodic} follow the same ordering as given in Theorem~(\ref{thm:main}).
 \begin{figure}[htbp]
   \centering
   \includegraphics[width=0.4\textwidth]{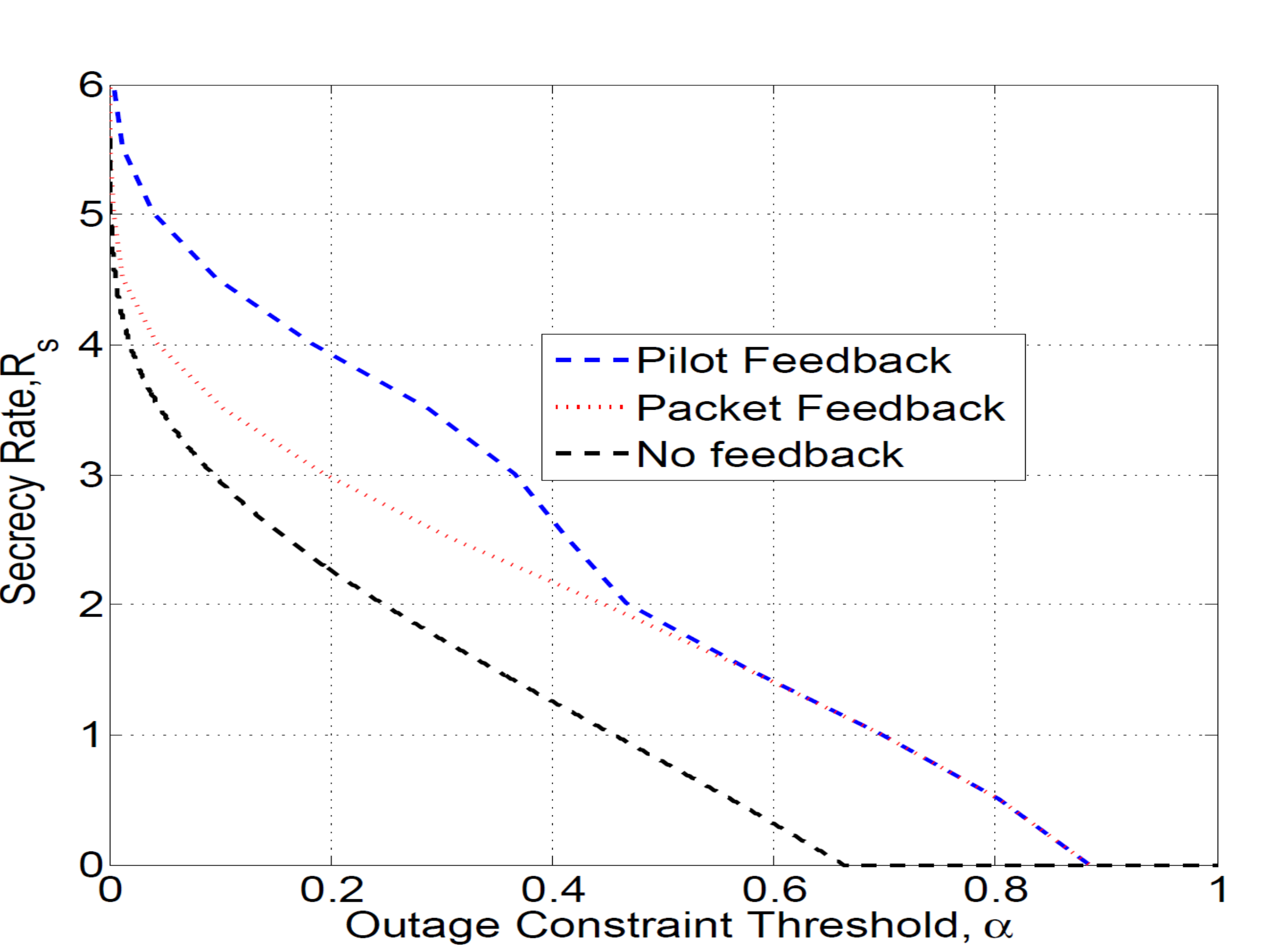}
   % where an .eps filename suffix will be assumed under latex,
   % and a .pdf suffix will be assumed for pdflatex
   \caption{Delay Limited Scenario: Comparison of CSI feedback methods under the outage constraint.}
   \label{fig:limited}
 \end{figure}

Next, we simulate the ergodic scenario and compare the two strategies, encoding across blocks without CSI and block-by-block encoding with packet CSI feedback. We used the same power parameters as in the simulations for the delay-limited case. All three channels are assumed to be block Rayleigh-fading wtih $E[H_z] =1$
%We assume all the channel power gains are distributed with exponential distribution and the mean of the distribution of the jamming channel's power gain is equal to 1. Transmitter power, $P$, and jamming power, $P_j$ are equal to 10 dB and 5 dB, respectively. 
We select the encoding parameter, $h_z^{*}$ such that $P[H_z \leq h_z^{*}] = 0.75$.

In Figure~\ref{fig:ergodic}, we illustrate the region where the encoding across blocks outperforms the block-by-block encoding on the $\left(E[H_e],E[H_m]\right)$ space. The region to the left of the border, given in the plot contains the set of $\left(E[H_e],E[H_m]\right)$ for which the encoding across blocks results in a higher secrecy rate. The intuition behind this observation is that, when $E[H_m]$ is much larger than $E[H_e]$, the positive operator inside the upper bound, $R^{+}$ loses its significance. 
\begin{figure}[htbp]
   \centering
   \includegraphics[width=0.4\textwidth]{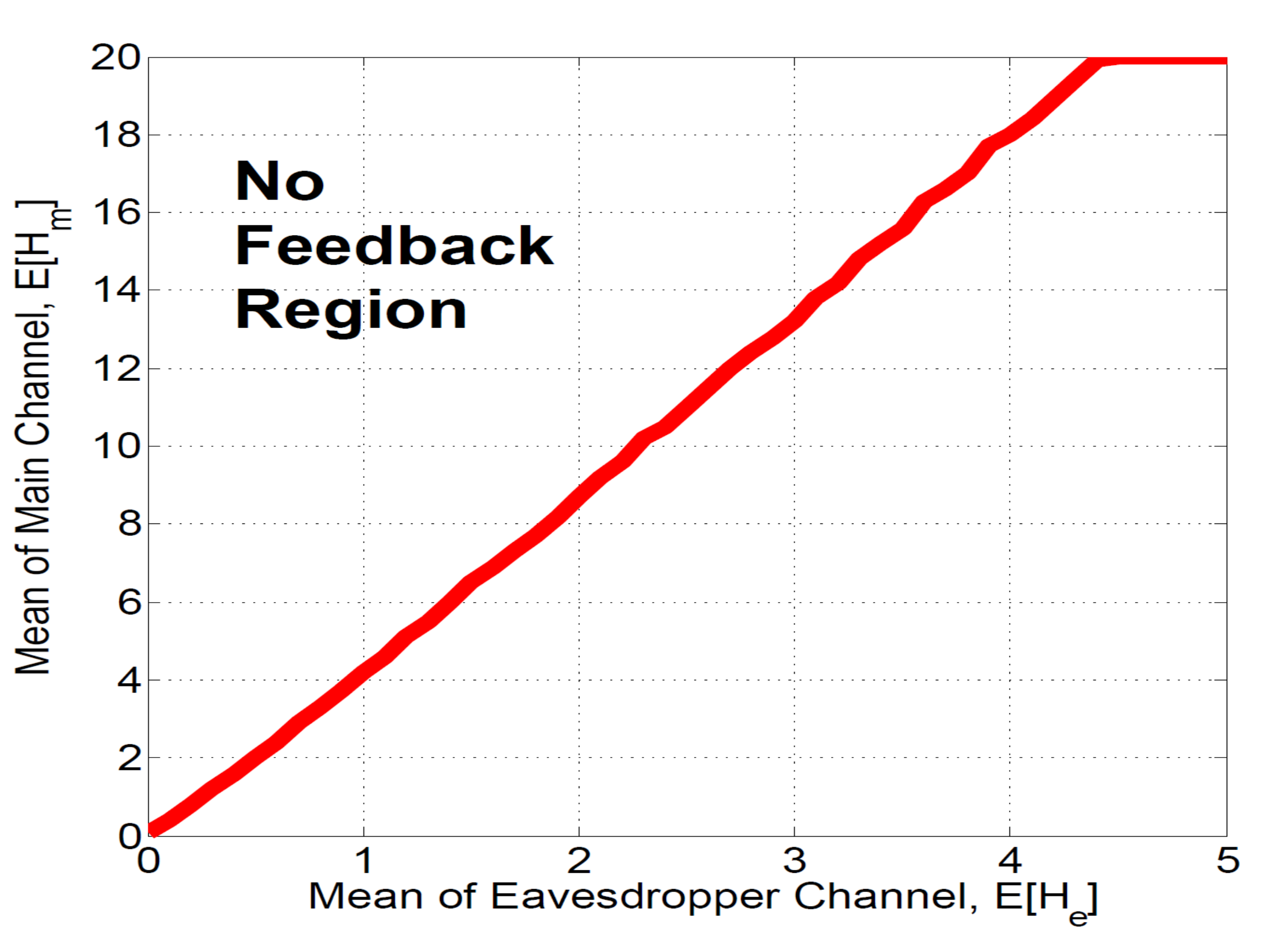}
   % where an .eps filename suffix will be assumed under latex,
   % and a .pdf suffix will be assumed for pdflatex
   \caption{The region where encoding across blocks with no CSI outperforms block-by-block encoding with packet feedback.}
   \label{fig:ergodic}
 \end{figure}
\section{Conclusion}
\label{chap:conc}
We consider the wiretap channel model under the presence of half duplex adversary that is capable of either jamming or eavesdropping at a given time. We analyzed the achievable rates under a variety of scenarios involving different methods for obtaining transmitter CSI. In particular, we considered no CSI, CSI with packet based feedback, and CSI with pilot based feedback. Each method provides a different grade of information not only to the transmitter on the main channel, but also to the adversary on all channels. We show for the delay limited scenario that, the highest secrecy rate is achieved with the pilot based feedback. Similarly, in the ergodic case, we prove that the pilot-based CSI feedback outperforms the packet-based CSI feedback, however interestingly, in certain cases no CSI may lead to a higher achievable secrecy rates than with CSI.
\section{Proof of Theorem~\ref{thm:main}}
\label{chap:proof}
Feasible set for the problem (\ref{problem}) is defined to be
\begin{equation}
\mathcal{F} = \{(R_s, R(\cdot)) : C \geq \alpha \text{,   } \forall A_p\}
\end{equation}
\noindent where $A_p$ is the set of channel power gains such that the adversary is in the jamming state if $h_A(i)\in A_p$. %$C$ is the function of both $A_p$ and $(R_s, R(\cdot))$ as seen in (\ref{constraint}). 
We have the following lemma.
\begin{lemma}
If $\mathcal{F}\neq \emptyset$, the solution to (\ref{problem}) is identical for all strategies of the adversary.
\end{lemma}
\begin{proof}
The equivalent form of (\ref{problem}) is as follows:
\begin{equation}
R_s^{*} = \max_{(R_s, R(\cdot))\in\mathcal{F}} R_s \label{equivalent}
\end{equation}
As seen from (\ref{equivalent}), $R_s^{*}$ does not depend on strategy of the adversary.
\end{proof}
However, the size of feasible set directly depends on the CSI feedback scheme. From (\ref{equivalent}), we can see that the solution to (\ref{problem}) is  directly proportional with the size of the feasible set. In the rest of the proof, we will use this observation to show the ordering.
\begin{lemma}
$R_s^{\text{No CSI}}\leq R_s^{\text{Packet Feedback}}$
\end{lemma}
\begin{proof}
In the no feedback case, constraint term (\ref{constraint2}) is reduced to
\allowdisplaybreaks
\begin{align}
&C=\lim_{M\to \infty} 1/M\sum_{i=1}^M I_J (i) I_{\log\left(1+\frac{PH_m(i)}{1+P_j H_z(i)}\right) \geq R}\nonumber\\
& +(1-I_J (i)) I_{\log\left(1+PH_e(i)\right) \leq R-R_s}I_{\log\left(1+PH_m(i)\right) \geq R}\\
&=\lim_{M\to \infty} 1/M\sum_{i=1}^M I_{H_e(i)\in A_p} I_{\log\left(1+\frac{PH_m(i)}{1+P_j H_z(i)}\right) \geq R}\nonumber\\
&+I_{H_e(i)\not\in A_p} I_{\log\left(1+PH_e(i)\right) \leq R-R_s}I_{\log\left(1+PH_m(i)\right) \geq R} \label{noCSI}\\
&= E\left[I_{H_e\in A_p} I_{\log\left(1+\frac{PH_m}{1+P_j H_z}\right) \geq R}\right.\nonumber\\
&\left.+I_{H_e\not\in A_p} I_{\log\left(1+PH_e\right) \leq R-R_s}I_{\log\left(1+PH_m\right) \geq R}\right] \text{, w.p. 1}\label{SLL}\\
&=\int_{h_e} E\left[I_{\log\left(1+\frac{PH_m}{1+P_j H_z}\right) \geq R} I_{H_e \in A_p}\right.\nonumber\\ 
&\;+\left.I_{\log\left(1+PH_e\right) \leq R-R_s}I_{\log\left(1+PH_m\right) \geq R} I_{H_e \notin A_p}\mid H_e=h_e\right]\nonumber\\
&\qquad\qquad\qquad \qquad  \qquad\qquad\qquad f_{H_e}(h_e)\,d h_e \text{, w.p. 1}\\
&=\int_{h_e} E\left[I_{\log\left(1+\frac{PH_m}{1+P_j H_z}\right) \geq R}\right] I_{h_e \in A_p}\nonumber\\  
&\qquad\qquad+E[I_{\log\left(1+PH_m\right) \geq R}] I_{\log\left(1+Ph_e\right) \leq R-R_s} I_{h_e \notin A_p}\nonumber \\
&\qquad\qquad \qquad \qquad\qquad \qquad \qquad f_{H_e}(h_e)\,d h_e \text{, w.p. 1} \label{strategy}
\end{align}
where  (\ref{noCSI}) follows from the fact that we have $I_J (i)=I_{H_e(i)\in A_p}$ since the adversary only knows the instantaneous power gain of the eavesdropper channel, (\ref{SLL}) follows  from the strong law of large numbers theorem and   (\ref{strategy}) follows from the independence of $H_e$, $H_m$, and $H_z$.

\noindent$A_p^*$ set that minimizes $C$ is as follows:
\begin{equation}
A_p^{*} = \{h_e \; :  \log\left(1+Ph_e\right)\leq R-R_s\}
\end{equation}
The minimized $C$ for a given rate pair ($R_s,R$) is
\begin{align}
&C^{\text{No CSI}}=\nonumber\\
&\mathbf{P}\left[R_s +\log\left(1+PH_e\right)\leq R\leq \log\left(1+\frac{PH_m}{1+P_j H_z}\right)\right]\nonumber
\end{align}
 
\noindent Feasible set for the transmitter for the no CSI feedback case is as follows
\begin{equation}
F^{\text{No CSI}} = \{ (R_s,R) \; :  C^{\text{No CSI}} \geq \alpha\}
\end{equation}
We now show that $F^{\text{No CSI}} \subset F^{\text{Packet Feedback}}$. For the packet feedback case, channel encoding rate is a function of the main CSI, $H_m$. The adversary knows the instantaneous power gains of the eavesdropping channel and the main channel. Constraint term in (\ref{constraint2}) is reduced to
\allowdisplaybreaks
\begin{align}
&C=\lim_{M\to \infty}1/M \sum_{i=1}^M I_J (i) I_{\log\left(1+\frac{PH_m(i)}{1+P_j H_z(i)}\right) \geq R(H_m(i))}\nonumber \\  
&+(1-I_J (i)) I_{\log\left(1+PH_e(i)\right) \leq R(H_m(i))-R_s}\nonumber\\
&\qquad\qquad\qquad\qquad\qquad\qquad I_{\log\left(1+PH_m(i)\right) \geq R(H_m(i))}\\
&=\lim_{M\to \infty}1/M \sum_{i=1}^M I_{\left(H_m(i),H_e(i)\right)\in A_p} I_{\log\left(1+\frac{PH_m(i)}{1+P_j H_z(i)}\right) \geq R(H_m(i))}\nonumber\\ 
&\qquad+I_{\left(H_m(i),H_e(i)\right)\not\in A_p} I_{\log\left(1+PH_e(i)\right) \leq R(H_m(i))-R_s} \nonumber\\
&\qquad\qquad\qquad\qquad\qquad\qquad I_{\log\left(1+PH_m(i)\right) \geq R(H_m(i))}  \label{PacketCSI}\\
&= E\left[I_{H_m,H_e\in A_p} I_{\log\left(1+\frac{PH_m}{1+P_j H_z}\right) \geq R(H_m)}\right.\nonumber\\
&\left. +I_{H_m,H_e\not\in A_p} I_{\log\left(1+PH_e\right) \leq R(H_m)-R_s}I_{\log\left(1+PH_m\right) \geq R}\right],\nonumber\\
&\qquad\qquad\qquad\qquad\qquad\qquad\qquad\qquad\qquad \text{w.p. 1}\label{SLL2}\\
%&=\int_{h_e} E\left[I_{\log\left(1+\frac{PH_m}{\sigma^2+P_j H_z}\right) \geq R(H_m)} I_{H_m,H_e \in A_p} + I_{\log\left(1+\frac{PH_e}{\sigma^2}\right) \leq R(H_m)-R_s} I_{H_m,H_e \notin A_p}\mid H_e=h_e\right]\nonumber\\
%&f_{H_e}(h_e)f_{H_m}(h_m)\,d h_e d h_m \text{, w.p. 1}\\
&=\int_{h_m,h_e} E\left[I_{\log\left(1+\frac{Ph_m}{1+P_j H_z}\right) \geq R(h_m)}\right] I_{h_m,h_e \in A_p} \nonumber\\
& + I_{\log\left(1+Ph_e\right) \leq R(h_m)-R_s} I_{\log\left(1+Ph_m\right) \geq R(h_m)}I_{h_m,h_e \notin A_p} \nonumber\\
&\qquad\qquad \qquad f_{H_e}(h_e)f_{H_m}(h_m)\,d h_e d h_m \text{, w.p. 1} \label{strategy2}
\end{align} 
where (\ref{PacketCSI}) follows from the facet that $I_J (i)=I_{H_m(i),H_e(i)\in A_p}$ since the adversary knows both the instantaneous power gain of the eavesdropper channel and the main channel. (\ref{SLL2}) follows from the strong law of large numbers theorem and (\ref{strategy2}) follows from the independence of $H_e$, $H_m$, and $H_z$.
\noindent$A_p^*$ set that minimizes $C$ is as follows:
\begin{align}
A_p^{*} &= \{h_m,h_e \; :  \log\left(1+Ph_e\right)\leq R(h_m)-R_s\}\nonumber\\ 
             &\bigcup \{h_m,h_e \; :  \log\left(1+Ph_m\right)\leq R(h_m)\}  \label{setmin2}       
\end{align}
When we combine (\ref{strategy2}) and (\ref{setmin2}), the minimized can be $C^{\text{Packet Feedback}}$ written as
\begin{align}
&C^{\text{Packet Feedback}}=\nonumber\\
& P\left[R_s +\log\left(1+PH_e\right)\leq R(H_m)\leq \log\left(1+\frac{PH_m}{1+P_j H_z}\right)\right]\nonumber.
\end{align}
 
\noindent Feasible set for the transmitter for the packet feedback case is as follows
\begin{equation}
F^{\text{Packet Feedback}} = \{ (R_s,R(\cdot)) \; :  C^{\text{Packet Feedback}} \geq \alpha\}
\end{equation}

\noindent It is easy to see that $F^{\text{No CSI}} \subset F^{\text{Packet Feedback}}$ then we have $R_s^{\text{No CSI}}\leq R_s^{\text{Packet Feedback}}$. 
\end{proof}
\begin{lemma}
$R_s^{\text{Packet Feedback}}\leq R_s^{\text{Pilot Feedback}}$
\end{lemma}
\begin{proof}
We will show that for any given $(R_s,R(\cdot))$,  $C^{\text{Packet Feedback}}$ is a lower bound for (\ref{constraint2}) in  the channel feedback cases. Constraint (\ref{constraint2}) can be written as
\begin{equation}
C=\lim_{M\to \infty} 1/M \sum_{i=1}^{M} \left[1- I_{H_m(i),H_e(i),H_z(i) \in \mathcal{O}}\right]\label{newconst}%= \lim_{M\to \infty}1/M \sum_{i=1}^{M} \left[ I_{(h_m(i),h_e(i),h_z(i) \in \overline{\mathcal{O}}}\right] \label{newconst}
\end{equation}
where  \allowdisplaybreaks
\begin{align}
&\mathcal{O}=\{h_m(i),h_e(i),h_z(i): I_J (i) I_{\log\left(1+\frac{Ph_m(i)}{1+P_j h_z(i)}\right) \leq R(h_m(i))} =1\}\nonumber\\
&\bigcup \{h_m(i),h_e(i),h_z(i): I_E (i) I_{\log\left(1+Ph_e(i)\right) \geq R(h_m(i))-R_s}=1\} \nonumber \\
&\bigcup \{h_m(i),h_e(i),h_z(i): I_E (i) I_{\log\left(1+Ph_m(i)\right)\leq R(h_m(i)) }=1\} 
\end{align}
 Constraint $C$ decreases as the siz of set $\mathcal{O}$ increases. Let's define an upper bound for $\mathcal{O}$:  
\begin{align} \allowdisplaybreaks
&\mathcal{O}^{\text{Upper Bound}}\nonumber\\
&=\{h_m(i),h_e(i),h_z(i): I_{\log\left(1+\frac{Ph_m(i)}{1+P_j H_z(i)}\right) \leq R(h_m(i))} =1\}\nonumber\\
&\bigcup \{h_m(i),h_e(i),h_z(i): I_{\log\left(1+Ph_e(i)\right) \geq R(h_m(i))-R_s}=1\} \nonumber \\
&\bigcup \{h_m(i),h_e(i),h_z(i): I_{\log\left(1+Ph_m(i)\right)\leq R(h_m(i)) }=1\}
\end{align}
where 
$I_J(i)=I_E(i)=1, \forall i\in\mathbb{N}$. Then, we find a lower bound for $C$ by putting $\mathcal{O}^{\text{Upper Bound}}$ in (\ref{newconst}):
\begin{align}
&C^{\text{Lower Bound}}\nonumber\\
&=\lim_{M\to \infty} \sum_{i=1}^{M} I_{\log\left(1+\frac{PH_m(i)}{1+P_j H_z(i)}\right) \geq R(H_m(i)) \geq \log\left(1+PH_e(i)\right) +R_s}\nonumber\\
&= C^{\text{Packet CSI}} \label{slw3}
\end{align}
Since $C^{\text{Lower Bound}}=C^{\text{Packet Feedback}}$,  we have $\mathcal{F}^{\text{Lower Bound}} = \mathcal{F}^{\text{Packet Feedback}} \subset \mathcal{F}^{\text{Pilot Feedback}}$.
\end{proof}
\allowdisplaybreaks
\section{Proof of Theorem~\ref{thm:nocsi}} \label{section:nocsi}
We employ the \emph{encoding across block} strategy explained in Section~\ref{subsec:encoding}, where $R_m \triangleq \mathbf{E}\left[\log\left(1+\frac{PH_m}{1+P_jH_z}\right)\right]$. If $R_m<\mathbf{E}\left[\log\left(1+\frac{PH_m}{1+P_jH_z}\right)\right]$, the adversary prevents the reliable communication by jamming at every block. Note that in no CSI case, the adversary only obtains $h_e(i)$ so equivocation rate is defined as $\frac{H(W|Z^{NM}, h_e^M)}{NM}$.   Equivocation analysis for the \emph{encoding across block} is as follows. \allowdisplaybreaks
\begin{align}
&H(W|Z^{NM}, h_e^M) \nonumber \\
&= H(W, X^{NM}|Z^{NM},h_e^M)-H(X^{NM}|Z^{NM}, W, h_e^M)\nonumber\\
&  = H(X^{NM}|Z^{NM}, h_e^M)+H(W|X^{NM},Z^{NM}, h_e^M)\nonumber\\
&\qquad\qquad\qquad\qquad\qquad\qquad-H(X^{NM}|Z^{NM}, W, h_e^M)\nonumber\\
&\geq H(X^{NM}|Z^{NM}, h_e^M)+H(X^{NM}|Z^{NM}, W, h_e^M)\nonumber\\
&= H(X^{NM}|h_e^{M})-I(X^{NM}, Z^{NM}|h_e^M)\nonumber\\
&\qquad\qquad\qquad\qquad\qquad\qquad+H(X^{NM}|Z^{NM}, W, h_e^M)\nonumber\\
&\stackrel{(a)}{=}MN R_m - I(X^{NM}, Z^{NM}|h_e^M)\nonumber \\
&\qquad\qquad\qquad\qquad\qquad\qquad-H(X^{NM}|Z^{NM}, W, h_e^M)\nonumber\\
&\geq MNR_m- N\sum_{i=1}^{M}I(X^{N}(i), Z^{N}(i)|h_e(i))\nonumber\\
&\qquad\qquad\qquad\qquad\qquad\qquad-H(X^{NM}|Z^{NM}, W, h_e^M)\nonumber\\
&\geq MN R_m - N\sum_{i=1}^M \log(1+Ph_e(i))\nonumber\\
&\qquad\qquad\qquad\qquad\qquad\qquad -H(X^{NM}|Z^{NM}, W, h_e^M) \nonumber
\end{align}
where $(a)$ follows from the fact that codeword $X^{NM}$ is uniformly distributed over a set of size $2^{NMR_m}$. We continue with the following steps.
\begin{align}
&\frac{H(W|Z^{NM}, h_e^M)}{NM} \nonumber\\
&\geq R_m-\frac{\sum_{i=1}^M \log(1+Ph_e(i))}{M}-\frac{H(X^{NM}|Z^{NM}, W, h_e^M)}{NM}\nonumber\\
&\stackrel{(b)}{\geq}R_m- E\left[\log(1+PH_e)\right]-\epsilon_1\nonumber \\
&\qquad\qquad\qquad\qquad-\frac{H(X^{NM}|Z^{NM}, W, h_e^M)}{NM}, \nonumber\\
&\stackrel{(c)}{\geq}R_m- E\left[\log(1+PH_e)\right]-\epsilon_1-\epsilon_2\nonumber\\
&=R_s -\epsilon,\nonumber
\end{align}
where $\epsilon = \epsilon_1+\epsilon_2$. Here, $(b)$ is satisfied for any $\epsilon_1>0$ and $h_e^M\in A_M$ with $Pr[A_M]=1$ and $M\geq M(\epsilon_1)$ since $$\lim_{M\rightarrow \infty}\frac{1}{M}\sum_{i=1}^M\log(1+PH_e(i))=E\left[\log(1+PH_e)\right]$$\text{with probability 1}, $(c)$ follows from the Fano's inequality. Let's define $R_e\triangleq R_m-R_s$ and $P_e^{NM} \triangleq P[X^{NM} \neq \hat{X}^{NM}]$ where $\hat{X}^{NM}=g(Z^{NM},h_e^M, W)$ is the estimation of the codeword $X^{NM}$.
\begin{align}
&\frac{H(X^{NM}|Z^{NM}, W, h_e^M)}{NM} \leq P_e^{NM} R_e +\frac{H(P_e^{NM})}{NM}\\
&\qquad\qquad\qquad\qquad\qquad                 \leq \epsilon_2 \label{fano}
\end{align}
Here, any $\epsilon_2>0$, (\ref{fano}) is satisfied for sufficiently high $N$ and $M$. The reason is that since $R_e=I(X^{N},Z^{N}|H_e)$, $P_e^{NM} \to 0$ for the sequence of codes ($2^{NMR_e}, R_e$) as $M \to \infty$.
% Use \section* for acknowledgement
\section{Proof of Theorem~\ref{upperbound}}
\label{chap:upperbound}
We first show that if $R_s$ is an achievable secrecy  rate under the packet feedback strategy, we have 
for any $\epsilon > 0, \frac{1}{NM}H(W|Z^{NM},H_m^M, H_e^M) \geq R_s-\epsilon, \forall N\geq N(\epsilon), \forall M\geq M(\epsilon)$. Note that here, the message $W$ is conditioned on random vectors, $H_e^M$ and $H_m^M$.
\allowdisplaybreaks
\begin{align}
& \frac{1}{NM} H(W| Z^NM, H_{m}^M, H_{e}^M) \\
&= \int_{ \mathcal{A}_M}\frac{1}{NM} H(W| Z^NM, h_{M}^M, h_{e}^M) f_{H_m^M,H_e^M} (h_m^M, h_e^M) \;dh_m^M \;dh_e^M\nonumber\\
&\geq \int_{\mathcal{A}_M}(R_s -\epsilon) f_{H_m^M,H_e^M} (h_m^M, h_e^M) \;dh_m^M \;dh_e^M\label{definition}\\
&= R_s-\epsilon \label{pr1}
\end{align}
where 
$\mathcal{A}_M$ is the set defined in Section~\ref{chap:ergodicformulation}. Here, (\ref{definition}) follows from the definition of achievability, and (\ref{pr1})follows from the fact that $P[\mathcal{A}_M]=1$.
We consider the following case for the rest of the proof.
When the transmitter receives a NACK signal, on the next block, the transmitter sends an independent group of bits instead of retransmitting the previous packet~\cite{yara2011}. In~\cite{yara2011}, the authors use this scheme for the secret key sharing. The crucial observation is that an upper bound for an achievable secret key rate is also an upper bound for achievable secrecy rates.

We define index set F that contains the indexes of blocks on which the transmitted codeword is successfully decoded. Suppose that size of F is $M'$, $n' \triangleq NM'$, and $n \triangleq NM$. We  now  prove that $\frac{1}{n}H(W|Z^n,H_m^M, H_e^M) \leq R_s^+ \text{ as  } n\to \infty$. 
\allowdisplaybreaks
\begin{align}
&H(W|Z^n,H_m^M, H_e^M) \nonumber\\
  &\stackrel{(a)}{\leq}H(W|Z^n,H_m^M, H_e^M) \nonumber\\
&\;\;\;-H(W|Z^n,Y^n,H_m^M,H_z^M, H_e^M)+n\delta_n\nonumber\\
     &\stackrel{(b)}{=} H(W|Z^{n'},H_m^{M'}, H_e^{M'})\nonumber\\
&\;\;\;-H(W |Z^{n'},Y^{n'},H_m^{M'},H_z^{M'}, H_e^{M'})+n\delta_n  \nonumber\\
&= I(W;Y^{n'}, H_z^{n'}|Z^{n'},H_m^{M'}, H_e^{M'})+\delta_n \nonumber\\
&\stackrel{(c)}{\leq} I(X^{n'};Y^{n'}, H_z^{M'}|Z^{n'},H_m^{M'}, H_e^{M'})+\delta_n \nonumber\\
&= I(X^{n'}; H_z^{M'}|Z^{n'},H_m^{M'}, H_e^{M'}) \nonumber\\
&\;\;\;+I(X^{n'};Y^{n'} |Z^{n'},H_m^{M'}, H_e^{M'}, H_z^{M'}) +n\delta_n \nonumber\\
&\stackrel{(d)}{=}I(X^{n'};Y^{n'} |Z^{n'},H_m^{M'}, H_e^{M'}, H_z^{M'})  \nonumber\\
&\stackrel{(e)}{=}\sum_{i=1}^M I(X^N(i);Y^N(i)|Z^N(i), H_z(i),H_m(i),H_e(i)) I(i)\nonumber \\
&\qquad \qquad \qquad \qquad \qquad \qquad\qquad \qquad \qquad \qquad  +n\delta_n \nonumber\\
&\stackrel{(f)}{\leq}\sum_{i=1}^M N\mathbf{E}\left[\log\left(1+\frac{PH_m}{1+P_jH_z}\right)-\log(1+PH_e)\right]^+ I(i)\nonumber\\
& \qquad \qquad \qquad \qquad \qquad \qquad\qquad \qquad \qquad \qquad  +n\delta_n \nonumber
\end{align}
where $I(i)= I_{H_z(i)\leq h_z^{*}}$. Here, (a) follows from Fano's inequality (b) follows from the independent choice of the codeword
symbols transmitted in each block that does not
allow the eavesdropper to benefit from the observations corresponding to the previous NACKed blocks, (c) results from the data processing inequality, (d) follows from the independence of $X^{n'}$ and $H_z^{M'}$, (e) follows from~\cite{lai2008}, and (f) follows from~\cite{leung1978}. 
\begin{align}
 &\frac{1}{n}H(W|Z^n,H_m^M, H_e^M)\\
      &\leq \mathbf{E}\left[\log\left(1+\frac{PH_m}{1+P_jH_z}\right)-\log(1+PH_e)\right]^+\nonumber\\
& \qquad \qquad \qquad \qquad \qquad \qquad \times\frac{1}{M}\sum_{i=1}^M \mathbf{I(i)}+\delta_n\\
&R_e \stackrel{(g)}{\leq}
\alpha\mathbf{E}\left[\log\left(1+\frac{PH_m}{1+P_jH_z}\right)-\log(1+PH_e)\right]^+\nonumber
\end{align}
(g) follows from the fact that $\delta_n \to 0 \text{ as } N,M \to \infty$ and from the ergodicity of channels such that  $\frac{1}{M}\sum_{i=1}^M I(i) \to P[H_z\leq h_z^{*}] \text{ as } M \to \infty$.
\section{Proof of Theorem~\ref{comparison}}
\label{chap:comparison}

%In Section~\ref{chap:numeric}, we compare $R_s^{+}$ and $R_s^{\text{No Feedback}}$.
Suppose that $R_s$ is a secrecy rate achieved with the packet based strategy and $n\triangleq NM$. Notice that the equivocation rates for the pilot and packet feedbacks are defined as $\frac{1}{n}H(W|Z^n, h_e^M)$ and $\frac{1}{n}H(W|Z^n, h_e^M, h_m^M)$, respectively.  Since $R_s$ is an achievable rate with the packet based strategy, by definition, for any $\epsilon > 0$ there exists $N(\epsilon)$ and $M(\epsilon)$ such that for $N \geq N(\epsilon)$ and $M \geq M(\epsilon)$, we have  $\frac{1}{n}H(W|Z^n,h_e^M, h_m^M) \geq R_s-\epsilon$,    $\forall (h_e^M, h_m^M) \in \mathcal{A}_M$  with $P(\mathcal{A}_M) =1$.

We define $\mathcal{A}_M(h_e^M) = \{h_m^M: (h_m^M, h_e^M)\in \mathcal{A}_M\}$. Since $H_m^M$ and $H_e^M$ are independent random vectors and $P[ (H_m^M, H_e^M) \in\mathcal{A}_M] =1$, we have $P[H_m^M \in \mathcal{A}_M(h_e^M)] =1, \forall h_e^M \in \mathcal{A}_M$. To observe that 
\begin{align}
1 &= \int_{\mathcal{A}_M} f_{H_m^M,H_e^M} (h_m^M, h_e^M) \;dh_m^M \;dh_e^M\\
   &= \int_{h_e^M } f_{H_e^M} (h_e^M)\int_{h_m^M \mathcal{A}_M(h_e^M)} f_{H_m^M} (h_m^M) \;dh_m^M \;dh_e^M\\
   & = \int_{h_e^M } P[ H_m \in \mathcal{A}_M(h_e^M)] f_{H_e^M} (h_e^M) \;dh_e^M =1
\end{align}
We can see that $P[ H_m \in \mathcal{A}_M(h_e^M)] =1$,  $\forall h_e^M \in \mathcal{E}$ such that  $P[\mathcal{E}]=1$. We now prove the lemma with following inequalities.
\begin{align}
&\frac{1}{n}H(W|Z^n, h_e^M)  \stackrel{(a)}{\geq} \frac{1}{n}H(W|Z^n,h_e^M, H_m^M)\\
&=\int_{\mathcal{A}_M(h_e^M)} \frac{1}{n}H(W|Z^n,h_e^M, h_m^M) f_{H_m^M}(h_m^M) \;dh_m^M \\
&\stackrel{(b)}{\geq}\int_{\mathcal{A}_M(h_e^M)}  (R_s-\epsilon)  f_{H_m^M}(h_m^M)   \;dh_m^M \\
&\stackrel{(c)}{=} R_s-\epsilon,\qquad \forall h_e^M \in \mathcal{E} \text{   with  } P[\mathcal{E}]=1.
\end{align}
$(a)$ follows from the fact that  conditioning reduces the entropy, $(b)$ follows from the fact that  since $h_m^M \in \mathcal{A}_M(h_e^M)$,  $(h_m^M, h_e^M) \in \mathcal{A}_M$, and $(c)$ follows from the fact that $P[ H_m \in \mathcal{A}_M(h_e^M)] =1$.

%% References:
%% We recommend the usage of BibTeX:
%%
%\bibliographystyle{IEEEtran}
%\bibliography{definitions,bibliofile}
%%
%% where we here have assume the existence of the files
%% definitions.bib and bibliofile.bib.
%% BibTeX documentation can be obtained at:
%% http://www.ctan.org/tex-archive/biblio/bibtex/contrib/doc/
%%
%%
%%
%% Or manual references (pay attention to consistency!):

\end{document}